\documentclass[reqno]{amsart}
\usepackage{enumitem}
\usepackage[all]{xy}
\usepackage{newlfont}
\usepackage[centertags]{amsmath}
\usepackage{amsfonts}
\usepackage{amssymb}
\usepackage{plain}
\usepackage{mathrsfs}

\newtheorem{thm}{Theorem}
\newtheorem{lem}{Lemma}

\newtheorem{defn}{Definition}

\newtheorem{exmp}{Example}

\begin{document}
\title[The cosmological time functions and lightlike rays]{The cosmological time functions and lightlike rays}
\author{Fatemeh Koohestani}

\author{Neda Ebrahimi}

\author{Mehdi Vatandoost}

\author{Yousef Bahrampour}

\email{koohestanifateme@gmail.com}
\email{n$\underline{  \:\:}$ebrahimi@uk.ac.ir}
\email{m.vatandoost@hsu.ac.ir}
\email{bahram@uk.ac.ir}
\date{}
\dedicatory{\textit{\normalsize{$^{1, 2,3}$Department of Pure Mathematics, Faculty of Mathematics and Computer, Shahid Bahonar University of Kerman, Kerman, Iran,\\
$^{4}$Department of Mathematics and Computer Sciences, Hakim Sabzevari University.}}}
\maketitle
\begin{abstract}
It is proved that all discontinuity points of a finite cosmological time function, $\tau$, are on past lightlike rays. As a result, it is proved that if $(M,g)$ is a chronological space-time without past lightlike rays then there is a representation of $g$ such that its cosmological time function is regular. In addition, by reducing conditions of regularity sufficient conditions for causal simplicity and causal pseudoconvexity of space-time is given. It is also proved that the second condition of regularity can be reduced to satisfies only on inextendible past-directed causal rays if $(M,g)$ be a space-time, conformal with an open subspace of Minkowski space-time or $\tau$ be continuous.
\end{abstract}


\section{\label{sec:level1}Introduction}
\noindent The concept of cosmological time function was defined in \cite{AG}. Time functions are important in study of the global causal theory of space-time. They can be defined arbitrary and may have little physical significance. But the cosmological time function is defined canonically and consequently, study of it gives us important information about space-time.
Let us recall its definition.\\
Let $(M,g)$ be a space-time and $d:M\times M\rightarrow [0,\infty]$ be the Lorentzian distance function. The cosmological time function $\tau:M\rightarrow [0,\infty)$ is defined by:
$$\tau(q):=sup_{p\leq q} d(p,q),$$
\noindent where $p\leq q$ means that $q\in J^{+}(p)$.\\
A time function on the space-time $(M,g)$ in the usual sense is a real valued continuous function which is strictly increasing on causal curves. The existence of such a function on $M$ requires causal stability (no closed causal curve in any Lorentzian metric sufficiently near the space-time metric exists). Although $\tau$ is not always well behaved, if it is regular then it is a Cauchy time function \cite{AG}.
\begin{defn}\cite{AG} The cosmological time function $\tau$ of $(M,g)$ is regular if and only if:
\begin{itemize}
\item $\tau (q)<\infty $, for all $q\in M,$
\item $\tau\rightarrow 0$ along every past inextendible causal curve.
\end{itemize}
\end{defn}
\noindent The first condition is an assertion that for each point $q$ of the space-time any particle that passes through $q$ has been in existence for a finite time (the space-time has an initial singularity in the strong sense). The second condition asserts that every particle came into existence at the initial singularity.\\
\noindent it is proved \cite{AG} that if the cosmological time function of $(M,g)$ is regular then:
\begin{itemize}
\item it is an almost every where differentiable time function;
\item $(M,g)$ is globally hyperbolic;
\item For each $q\in M$ there is a future-directed timelike ray $\gamma_{q}:(0,\tau(q)]\rightarrow M$ that realizes the distance from the initial singularity to $q$, that is $\gamma_{q}$ is future-directed timelike unit geodesic which is maximal on each segment, such that:\\
\begin{equation}
\gamma_{q}(\tau(q))=q,~~~~\tau(\gamma_{q}(t))=t
,~~~~ for ~t\in (0,\tau(q)].
\end{equation}
\end{itemize}
We recall that since the Lorentzian distance is not conformally invariant, $\tau$ is not too and even if $(M,g)$ is globally hyperbolic, $\tau$ is not necessarily regular.
But it is proved in \cite{NE} that if $(M,g)$ is globally hyperbolic then there is a smooth real function $\Omega >0$ such that $\tau$ is regular, on $(M,\Omega g)$. This gives a characterization for global hyperbolicity.
\\ It is also proved \cite{CH} that the regular $\tau$ is differentiable at $q$ iff there exists only one timelike ray of the form (1).\\
In the whole of this paper we suppose that $\tau$ is finite.
We reduce the second condition of regularity and investigate its effects on the cosmological time function and causal properties of space-time.\\
In the second section, we recall some definitions and theorems which are used widely in this paper.
In the third section of this paper it is proved that a point of discontinuity appears on a past lightlike ray. Hence $\tau$ is a time function in space-times with no past lightlike rays. Using this time function it can be shown by a simple proof that non-totally vicious space-times without lightlike rays are globally hyperbolic. A different proof was given in \cite{MC}.\\
In addition, it can be asked that is it possible to reduce the second condition in the following way:\\
$\tau\rightarrow 0$, along every past inextendible causal geodesic.\\
The following example shows that this is not the case in general.
\begin{exmp} \cite{AG} Let $M:=\{(x,y,t)\in S^1 \times \mathbb{R} \times \mathbb{R} : t>-1\}$ with the metric
\begin{center}
$ g:= dy^2 +e^{2y}(dxdt+(\vert t \vert^{2\alpha} +(e^{y^{2}}-1))dx^2 ).$
\end{center}
\end{exmp}
Although $\tau$ is going to zero along past inextendible causal geodesics it does not go to zero along causal curves.\\
In section 4, we investigate some conditions on $(M,g)$ and $\tau$ which imply that $\tau$ is regular in this situation. In addition, Using the cosmological time function sufficient conditions for causal simplicity and causal pseudoconvexity of space-time is given.\\
\section{preliminaries}
The standard notations from Lorentzian geometry are used in this paper. The reader is referred to \cite{B, P}. We denote with $(M,g)$ a $C^{\infty}$ space-time (a connected, Hausdorff, time orientted Lorentzian manifold) of dimension $n\geq 2$ and signature $(-,+,...,+)$. If $p, q\in M$, then $q\in I^{+}(p)$ (resp. $q\in J^{+}(p)$) means that there is a future directed timelike (resp. causal) curve from $p$ to $q$. $I^{+}(p)$ is called the chronological future and $J^{+}(p)$ the causal future of $p$. Likewise $I^{-}(p)$ and $J^{-}(p)$ are defined and are called chronological and causal past of $p$. $(M,g)$ is causal (resp. chronological) if there is no closed causal (resp. timelike) curve in it. A causal space-time is globally hyperbolic if $J^{+}(p)\cap J^{-}(q)$ be compact, for every $p, q\in M$ and is causally simple if $J^{\pm}(p)$ be closed, for every $p\in M$. In addition, a space-time is strongly causal if every point of it has arbitrary small causally convex neighbourhoods. The interested reader is referred to \cite{M2} for more details.\\
If $q\in J^{+}(p)$ the Lorentzian distance $d(p,q)$ is the supremum of the length of all causal curves from $p$ to $q$ and if $q\notin J^{+}(p)$ then $d(p,q)=0$. If $(M,g)$ is strongly causal then there is a conformal class of $g$ such that $d$ and $\tau$ are bounded.
\begin{lem}\cite{M}
Let $h$ be an auxiliary complete Riemannian metric on $M$ and let $\rho$ be the associated distance. Let $q\in M$ and $B_{n}(q)=\{r: \rho(q,r)<n\}$ be the open ball of radius $n$ centered at $q$. If $(M, g)$ is strongly causal, then there is a smooth function $\Omega>0$, such that $\textmd{diam}(M,\Omega g)=sup \{d(p,q), p,q\in M\}$ is finite and for every $\epsilon>0$ there is $n\in N$ such that if $\gamma:I\rightarrow M$ is any $C^{1}$ causal curve,
$$\int_{I\cap \gamma^{-1}(M-\overline{B_{n}(q)})}\sqrt{-g(\dot{\gamma},\dot{\gamma})}dt<\epsilon$$
that is, its many connected pieces contained in the open set $M-\overline{B_{n}(q)}$ have a total Lorentzian length less than $\epsilon$.
\end{lem}
The limit curve theorem will be used several times. The reader is referred to \cite{ML} for a strong formulation. All the curves are parametrized by $h$- arc length, where $h$ is a complete Riemannian metric. A past (future) ray in $(M,g)$ is a maximal past (future) inextendible causal geodesic $\gamma:[0,\infty)\rightarrow M$.\\ A sequence of causal curves $\gamma_{n}:[a_{n},b_{n}]\rightarrow M$ is called limit maximizing if:\\
$$L(\gamma_{n})\geq d(\gamma_{n}(a_{n}), \gamma_{n}(b_{n}))-\epsilon_{n},$$ where $\epsilon_{n}\rightarrow 0$.\\
The following lemma is used in section 3 and 4.
\begin{lem}\cite{G2} Let $z_{n}$ be a sequence in $M$ with $z_{n}\rightarrow z$. Let $z_{n}\in I^{+}(p_{n})$ with finite $d(z_{n},p_{n})$. Let $\gamma_{n}:[0, a_{n}]\rightarrow M$ be a limit maximizing sequence of causal curves with $\gamma_{n}(0)=z_{n}$ and $\gamma_{n}(a_{n})=p_{n}$. Let $\widehat{\gamma_{n}}:[0,\infty)\rightarrow M$ be any future extension of $\gamma_{n}$. Suppose either:
\begin{itemize}
\item $p_{n}\rightarrow \infty$, i.e. no subsequent is convergent,
\item $d(z_{n},p_{n})\rightarrow \infty$.
\end{itemize}
Then any limit curve $\gamma:[0,\infty)\rightarrow M$ of the sequence $\widehat{\gamma_{n}}$ is a causal ray starting at $z$.
\end{lem}

 \section{Continuity of cosmological time function and lightlike rays}
The cosmological time function $\tau$ is not a time function in general, but if $\tau<\infty$ it has the following property:\\
\begin{equation}
q\in J^{+}(p)\Rightarrow \tau(p)+d(p,q)\leq \tau(q).
\end{equation}
This implies that $\tau$ is isotone.
\begin{defn} A function $t:M\rightarrow R$ which satisfies $q\in J^{+}(p)\Rightarrow t(p)\leq t(q)$ is said to be isotone.
\end{defn}
It is proved in \cite{M2} that isotones are almost everywhere continuous and differentiable.
\begin{thm}
Every isotone function $f:M\rightarrow R$ on $(M,g)$ is almost everywhere continuous and almost everywhere differentiable. Moreover, it is differentiable at $ p\in M $ iff it is $G\widehat{a}teaux$- differentiable at p. Finally, if $x:I \rightarrow M $ is a timelike curve, the isotone function $f$ is upper/lower semi-continuous at $x_0 = x(t_0 )$ iff $f\circ x$ has the same property at $t_0$.
\end{thm}
Hence the cosmological time function is almost everywhere continuous. The question is where a point of discontinuity appears.
\begin{thm} Let $(M,g)$ be a space-time with finite cosmological time function $\tau$. If $\tau$ be discontinuous at $q\in M$ then it lies on a past lightlike ray.
\end{thm}
\begin{proof} Since $p\mapsto d(p,q)$, for every $q\in M$, is lower semi-continuous, $\tau$ is lower semi-continuous too. Hence it is not upper semi-continuous at $q$. Consequently, there is a sequence $\{q_{n}\}$ and $\epsilon>0$ that $q_{n}\rightarrow q$ and $\tau(q_{n})\geq \tau(q)+\epsilon$. Suppose that $\{p_{n}\}$ be a sequence in $M$ such that $d(p_{n},q_{n})\geq \tau(q_{n})-1/n$. (2) implies that $\tau(p_{n})+d(p_{n},q_{n})\leq\tau(q_{n})$. Since $\tau$ is finite we have $d(p_{n},q_{n})\leq \tau(q_{n})-\tau(p_{n})$ and consequently:\\
$\tau(q_{n})-1/n\leq d(p_{n},q_{n})\leq \tau(q_{n})-\tau(p_{n})$.\\
This implies that $\tau(p_{n})\rightarrow 0$. Let $\{\gamma_{n}\}$ be a limit maximizing sequence such that for each $n$, $\gamma_{n}:[0,a_{n}]\rightarrow M$, is a past-directed timelike curve, $\gamma_{n}(0)=q_{n},~\gamma_{n}(a_{n})=p_{n}$, i.e. $L(\gamma_{n}) > d(p_{n},q_{n})-\epsilon_{n}$ where $\epsilon_{n}\rightarrow 0$.\\
Let $\widehat{\gamma_{n}}: [0, \infty)\rightarrow M$ be any past inextendible extension of $\gamma_{n}$ and $\gamma:[0,a]\rightarrow M$ be the limit curve of $\widehat{\gamma_{n}}$ with $\gamma(0)=q$. $\{p_{n}\}$ diverges to infinity (if $p_{n}\rightarrow p$ then by lower semi continuity of $\tau$, $\tau(p)\leq liminf(\tau(p_{n}))$. This implies that $\tau(p)=0$, which is a contradiction). Lemma 2 implies that $\gamma$ is a causal ray. If $\gamma$ is not a null ray, it is timelike. As all the curves are parametrized by arc-length, $limsup( l(\gamma_{n}\vert_{ [0,b]}))\leq l(\gamma\vert_{[0,b]})$. There are $\delta,~b>0$ such that $l(\gamma\vert_{[0,b]})+\delta \leq \epsilon/2$. Consequently, $l(\gamma_{n}\vert_{[0,b]})\leq l(\gamma\vert_{[0,b]})+\delta\leq \epsilon/2$, for suffitiently large $n$.\\
$l(\gamma_{n}\vert_{[b,a_{n}]})=l(\gamma_{n})-l(\gamma_{n}\vert _{[0,b]})\geq \tau(q)+\epsilon/2 -1/n-\epsilon_{n}$.\\
Hence for sufficiently large $n$, $l(\gamma_{n}\vert_{[b,a_{n}]})>\tau(q)$. Since $\gamma$ is timelike, $\gamma_{n}(b)\in I^{-}(q)$, for suffitiently large $n$. Hence $\tau(q)\geq l(\gamma_{n}\vert_{[b,a_{n}]})>\tau(q)$ which is a contradiction.
\end{proof}
As an application of Theorem 2 the following two theorems can be proved.
\begin{thm}
Let $(M,g)$ be a space-time without past lightlike rays and $\tau< \infty$ then for each $q$ there is a future-directed unit speed timelike ray $\gamma_{q}:(0,\tau(q)]\rightarrow M$ which is maximal in each segment, such that:\\
$$\gamma_{q}(\tau(q))=q,~~~\tau(\gamma_{q}(t))=t,~~~t\in (0,\tau(q)].$$
\end{thm}
\begin{proof}
Let $q\in M$ and $\{p_{n}\}$ be a sequence such that $d(p_{n},q)\geq \tau(q)-1/n$. $p_{n}$ diverges to infinity as it is proved in the pervious theorem. Let $\gamma_{n}:[0,a_{n}]\rightarrow M$, $\gamma_{n}(0)=q$, $\gamma_{n}(a_{n})=p_{n}$ be a limit maximizing sequence of curves:\\
$\epsilon_{n}=d(\gamma_{n}(a_{n}),\gamma_{n}(0)- l(\gamma_{n}[0,a_{n}])), \epsilon_{n}\rightarrow 0.$\\
and $\widehat{\gamma_{n}}$ be a past inextendible extension of $\gamma_{n}$, for $n\in \mathbb{N}$. By using of Lemma 2 and assumption $\widehat{\gamma_{n}}$ converges to a timelike ray $\gamma:[0,\infty)\rightarrow M$. It suffices to prove that,\\ $d(\gamma(b),q)=\tau(q)-\tau(\gamma(b))$, for every $b\in [0,\infty).$\\
$l(\gamma_{n}([0,b]))=l(\gamma_{n})-l(\gamma_{n}[b,a_{n}])\geq \tau(q)-(1/n)-\epsilon_{n}-\tau(\gamma_{n}(b))$.\\ Since $\tau$ is continuous we have,\\
$d(\gamma(b),q)\geq l(\gamma[0,b])\geq limsup( l(\gamma_{n}([0,b])))\geq \tau(q)-limsup(\tau(\gamma_{n}(b)))=\tau(q)-\tau(\gamma(b))$,\\ and the proof is complete.
\end{proof}
The proof of the following theorem is similar to what is given for regular cosmological time functions in \cite{AG}
\begin{thm}
Let $(M,g)$ be a space-time without past lightlike rays. If $\tau< \infty$ then it is a time function.
\end{thm}
\begin{proof} If $q\in I^{+}(p)$ then by using of (2) it is clear that $\tau(p)< \tau(q)$. Assume that $q\in J^{+}(p)-I^{+}(p)$ then there is a lightlike geodesic from $p$ to $q$. Let $\gamma_{p}$ be the timelike ray to $p$ guaranteed by Theorem 3 and $x\in \gamma_{p}$. By cutting the corner argument near $p$ we have:\\
$d(x,q)>d(x,p)+d(p,q)$. Consequently, $\tau(q)-\tau(p)\geq d(x,q)>d(x,p)=\tau(p)-\tau(x)>0$ and the proof is complete.
\end{proof}
The following theorem was proved in \cite{MC}.
\begin{thm}
Non-totally vicious space-time $(M,g)$ with no lightlike rays is a globally hyperbolic space-time.
\end{thm}
It is proved in \cite{MC} that a chronological space-time without lightlike lines is stably causal and consequently has a time function. Then using this time function the above theorem is proved. In this paper, by using of Theorem 4 and the following lemmas it can be shown that there is a representation of $(M,g)$ such that the cosmological time function is regular.
\begin{lem}
Let $(M,g)$ be a chronological space-time without past (or future) lightlike rays then $(M,g)$ is strongly causal.
\end{lem}
\begin{proof}
Suppose by contradiction that $(M,g)$ is not strongly causal. Then there is $p\in M$ and a sequence of arbitrary small relatively compact neighbourhoods $U_{n}$, $n\in \mathbb{N}$, of $p$ which are not causally convex, i.e, for every $U_{n}$ there exist $p_{n}, q_{n}\in U_{n}$ and a causal curve $\gamma_{n}$ from $p_{n}$ to $q_{n}$ which are not contained in $U_{n}$, $p_{n}\rightarrow p$, $q_{n}\rightarrow q$ and $p=q$. The second part of the limit curve theorem, Theorem 3.1 \cite{ML}, implies that one of the following cases occur:\\
1) $\gamma_{n}$s are contained in a compact set. Since $M$ is chronological, the limit curve $\gamma_{q}$ is a closed maximal lightlike curve. The curve $\gamma$ making infinite rounds around $\gamma_{q}$ is a past lightlike ray which is a contradition.\\
2) $\gamma_{n}$s are not contained in a compact set. Then $\gamma_{p}\circ \gamma_{q}=\gamma$ is the limit curve of $\gamma_{n}$s. Since $\gamma$ is not a lightlike line, the chronology violating set is non empty. This is a contradiction since $\tau< \infty$.
\end{proof}

 Lemma 2 and Lemma 3 imply the following theorem.
\begin{thm}Let $(M,g)$ be a chronological space-time without past lightlike rays then there is a positive real function $\Omega$, that the cosmological time function of $(M,\Omega g)$ is regular.
\end{thm}
\begin{proof} $(M,g)$ is strongly causal by using of Lemma 3. Let $q_{0}\in M$ and $\Omega>0$ be given as in Lemma 1. It is clear that $\tau$ is finite. Suppose by contradiction that there is a past inextendible causal curve $\eta:[0,\infty)\rightarrow M$, $\eta(0)=p$ such that $\tau\rightarrow a>0$, along $\eta$. Choose $n$ such that the length of any causal curve out of $\overline{B_{n}(q_{0})}$ is less than $\epsilon<a$ and $p\in B_{n}(q_{0})$. Since $M$ is strongly causal it is non-total imprisoning and consequently there is $t_{0}\in \mathbb{R}$ such that $\eta(t)\in M-\overline{B_{n}(q_{0})}$, for $t>t_{0}$. In addition suppose that $\{t_{n}\}$, $t_{n}\rightarrow \infty$, be a sequence of real numbers and $p_{n}=\eta(t_{n})$. Let $\gamma_{p_{n}}$ be the maximal timelike ray guaranteed by Theorem 3. Since $l(\gamma_{p_{n}})=\tau(p_{n})> \epsilon$, $\gamma_{p_{n}}\cap \overline{B_{n}(q_{0})}\neq \emptyset$. Again by using of non-total imprisoning condition, $\gamma_{p_{n}}$ escape $\overline{B_{n}(q_{0})}$ in a point $q_{n}=\gamma_{p_{n}}(s_{n})$. Since $\overline{B_{n}(q_{0})}$ is compact $q_{n}\rightarrow q$. let $\gamma'_{n}$ be a reparametrization of $\gamma_{p_{n}}\vert [0,s_{n}]$ with $h$ arc length in such a way that $\gamma'_{n}(t)=\gamma_{p_{n}}(s_{n}-t)$. $\gamma'_{n}\rightarrow \gamma$, $\gamma(0)=q$.\\
Since $\gamma_{p_{n}}$ is maximal, for every $n$, Lemma 2 implies that $\gamma$ is a timelike ray.\\
Indeed, we have:\\
$\tau(\gamma(b))-\tau(\gamma(0))=limsup(\tau(\gamma'_{n}(b))-\tau(\gamma'_{n}(0))=limsup(l(\gamma'_{n}\vert [0,b]))\leq l(\gamma\vert [0,b])=d(\gamma(0),\gamma(b))\leq liminf(d(\gamma'_{n}(0),\gamma'_{n}(b))=\tau(\gamma(b))-\tau(\gamma(0))$,
Since $\tau$ is continuous. In addition, $\tau$ is a time function and consequently $\tau(\gamma(b))-\tau(\gamma(0))>0$.\\
Let $\eta_{n}= \gamma'_{n}\circ \eta$. $\gamma$ and $\eta$ are the limit curves of $\eta_{n}$ and non of them are lightlike ray. The limit curve theorem implies that
$\eta\subset I^{+}(q)$. This means that $M$ has a TIF, which is a contradiction since $M$ has no past lightlike ray and the boundary of a TIF is generated by past lightlike rays.
\end{proof}
A non- totally vicious space-time without lightlike rays is choronological \cite{ME} and Theorem 5 is easily given by using of Theorem 6.
\section{The cosmological time function and causal rays}
In this section, we reduce the second condition of regularity and investigate the causality properties of space-times which satisfy:
\begin{itemize}\label{a}
\item $\tau$ is finite;
\item $\tau\rightarrow 0$ on past inextendible causal geodesics (or on past null rays).
\end{itemize}
Example 1 shows that in this case $\tau$ is not necessarily regular. Indeed, it can be checked that it is not globally hyperbolic. Causal pseudoconvexity and causal simplicity are weaker conditions than global hyperbolicity. We will prove that in a reflective space-time the above conditions imply causal pseudoconvexity (or causal simplicity). The space-time $(M,g)$ is called causally pseudoconvex if for each compact set $K$ there exists a compact set $K'$ such that each geodesic with both end points in $K$ has its image in $K'$.
\begin{defn}\cite{V1} Assume $p_n \rightarrow p$ and $q_n \rightarrow q$ for distinct points $p$ and $q$ in space-time $M$. We say that space-time $M$ has causal limit geodesic segment property (LGS), if each pair $p_n$ and $q_n$ can be joined by a geodesic segment, then there is a limit geodesic segment from $p$ to $q$.
\end{defn}
we can define causal, null or maximally null (LGS) property by restricting the condition of the above definition to causal, null or maximally null geodesics, respectively.
Theorem 7, Theorem 8 and Lemma 4 are used to prove the main results of this section. The following theorem gives a characterization for pseudoconvexity.
\begin{thm}\cite{V1}
Let $(M,g)$ be a strongly causal space-time. Then it is (null or maximal null) causal pseudoconvex if and only if it has (null or maximally null) causal LGS property.
\end{thm}
\begin{lem} If the cosmological time function $(M,g)$ satisfies the following conditions:
\begin{itemize}
\item $\tau< \infty$;
\item $\tau \rightarrow 0$, along past lightlike rays,
\end{itemize}
then $(M,g)$ is non-totally imprisoning.
\end{lem}
\begin{proof} Since $\tau< \infty$, $(M,g)$ is chronological. Suppose by contradiction that there is a past inextendible causal curve which is imprisoned in a compact set $K$. Hence there is a lightlike line which is imprisoned in $K$ \cite{ME}. Since $K$ is compact there are $p_{i}\in M$, $i=1,...,n$, that $K\subset \cup I^{+}(p_{i})$ and consequently $\tau(x)\geq min(\tau(p_{i}))$, $i=1,...,n$, for all $x\in K$. But it is a contradiction to assumption.
\end{proof}

 \begin{thm}\cite{P}
Let $S \subset M$, and set $B =\partial I^{+}(S)$. Then if $x \in B$, there exists a null geodesic $\eta \subseteq B$ with future endpoint $x$ and which is either past-endless or has a past endpoint on $S$.
\end{thm}

 \begin{thm} \cite{V2} Assume that $(M,g)$ be a causal space-time, but not causally simple. Then there are (1) $p$, $q\in M$ such that $p$ has a future inextendible maximal null geodesic ray in $I^{-}(p)$ or (2) $q$ has a past inextendible maximal null geodesic ray in $\partial I^{+}(q)$.
Conversely, assume $(M,g)$ has a $p$ (resp. q) such that $\partial I^{-}(p)$ has a future null geodesic ray (resp. $\partial I^{+}(p)$ has a past inextendible maximal null geodesic ray), then $(M,g)$ is not causally simple.
\end{thm}
{\bf Remark.} We recall that $(M,g)$ is reflecting iff $p \in \overline{I^{-}(q)}\Leftrightarrow q\in \overline{I^{+}(p)}$, for $p,q\in M$.
\begin{thm}
Let $(M,g)$ be a reflecting space-time such that its cosmological time function $\tau$ has the following properties:
\begin{itemize}
\item $\tau$ is finite;
\item $\tau\rightarrow 0$, along every past lightlike ray $\gamma$;
\end{itemize}
then $(M,g)$ is causally simple.
\end{thm}
\begin{proof}
Lemma 4 implies that $(M,g)$ is causal. Suppose by contradiction that $(M,g)$ is not causally simple. Hence there are $p$, $q\in M$ such that $q\in \overline{J^{+}(p)}-J^{+}(p)$ and (1) or (2) in Theorem 9 occurs:\\
In the first case there is a past lightlike ray $\gamma$ which is contained $q$ and lies in $\partial I^{+}(p)$. Since $(M,g)$ is reflecting, there is a sequence $\{p_{n}\}$, $p_{n}\rightarrow p$, such that $p_{n}\in I^{-}(q)$. Lower semi continuity of cosmological time function implies that $\tau(p)\leq liminf \tau(p_{n})\leq \tau(q)$. $q$ can be choosen arbitrary on the causal ray $\gamma$ and consequently $\tau(\gamma)\geq \tau(p)$ which is a contradiction, since $\tau\rightarrow 0$ along it.\\
Again in the second case reflectivity of $(M,g)$ implies that $q$ lies in $\partial I^{+}(p)$, hence by Theorem 8 there is a past lightlike ray $\alpha$ with end point $q$ in $\partial I^{+}(p)$. Indeed, it can not be a null geodesic from $p$ to $q$ since $q\notin J^{+}(p)$. The proof in this case is similar to case (1).
\end{proof}

\begin{thm} Let $(M,g)$ be a reflective space-time such that:
\begin{itemize}
\item $\tau$ is finite,
\item $\tau\rightarrow 0$ along past inextendible causal geodesics,
\end{itemize}
then $(M,g)$ is causally pseudoconvex.
\end{thm}
\begin{proof}
Theorem 10 implies that $(M,g)$ is causally simple. Suppose by contradiction that $(M,g)$ is not causally pseudoconvex then Theorem 7 implies that there are sequences $\{p_{n}\}$, $\{q_{n}\}$ and geodesics $\gamma_{n}$ from $p_{n}$ to $q_{n}$ such that $p_{n}\rightarrow p$, $q_{n}\rightarrow q$ but there limit curve $\gamma$ with $\gamma(0)=q$ does not contain $p$. Hence limit curve theorem implies that $\gamma(t)\in J^{+}(p)$ and consequently $\tau(p)\leq \tau(\gamma(t))$ which is a contradiction, since $\gamma$ is an inextendible past geodesic.
\end{proof}

 \begin{thm} Let $(M,g)$ be a space-time which satisfies the following properties:
\begin{itemize}
\item $\tau< \infty$;
\item $\tau\rightarrow 0$, along past inextendible causal rays;
\item $\tau$ is continuous;
\end{itemize}
then $\tau$ is regular.
\end{thm}
\begin{proof} Suppose by contradiction that $\tau\nrightarrow 0$ along a past inextendible causal curve $\gamma:[0,\infty)\rightarrow M$. Let $p_{n}=\gamma(t_{n})$, where $t_{n}\rightarrow \infty$ and $\gamma_{n}:[0,a_{n}]\rightarrow M$, $\gamma_{n}(0)=p$ and $\gamma_{n}(a_{n})=p_{n}$ be a limit maximizing sequence.
Since by Lemma 4 $(M,g)$ is non-total imprisoning, $\{p_{n}\}$ escapes to infinity. Consequently, the limit curve of $\gamma_{n}$, $\eta:[0,\infty)\rightarrow M$, $\eta(0)=p$ is a causal ray. $\gamma_{n}(t)\in I^{+}(p_{m})$, for $m\geq n$. Hence $\eta\subset \overline{ I^{+}(\gamma)}$. Since $\tau>a$ on $ I^{+}(\gamma)$ and it is continuous, $\tau(\eta(t))\geq a$, for every $t$, which is a contradiction.
\end{proof}
\begin{thm} Let $(M,g)$ be a space-time of dimension $n+1$, conformally embedded as an open subset in Minkowski space-time $(E^{n,1},h)$, and its cosmological time function satisfies the following properties:
\begin{itemize}
\item $\tau< \infty$;
\item $\tau\rightarrow 0$, along past causal rays;
\end{itemize}
then $\tau$ is regular.
\end{thm}

 \begin{proof} Suppose that $\tau\rightarrow a$, $a\neq 0$, along a past inextendible causal curve $\gamma:[0,\infty)\rightarrow M$. Let $p_{n}$, $\gamma_{n}$ and $\eta$ be as in the pervious theorem.
It is trivial that $\eta\subseteq \overline{I^{+}(\gamma)}$. $\partial (I^{+}(\gamma))$ is an achronal boundary. $\tau> a$ on $ I^{+}(\gamma)$, but $\tau\rightarrow 0$ on $\eta$. Hence $\eta$ has to escape $I^{+}(\gamma)$ in a point $x=\eta(s)$ and $\eta\vert [s,\infty)\subseteq \partial I^{+}(\gamma)$. This implies that $\eta$ is not a timelike ray. In addtion, if $\eta'(0)=w$ then $<w,\alpha'>_{h}=0$, for any lightlike ray $\alpha$ in $\partial I^{+}(\gamma)$, since $M$ is an open subset of Minkowski space-time. Hence $<w,\eta'>_{h}=0$, which is a contradiction.
\end{proof}

 \end{document}